\documentclass{article}

\usepackage[T1]{fontenc}

\let\shortcite\cite

\usepackage{graphicx}

\let\shortcite\cite

\usepackage{algorithm}
\usepackage{algorithmic}

\usepackage{pifont}
\usepackage{latexsym}
\usepackage{amsfonts}
\usepackage{amsmath}
\usepackage{amssymb}
\usepackage{hyperref}
\usepackage{amsthm}
\usepackage{diagbox}

\newcommand{\rp}{{\rm RP}}
\newcommand{\p}{{\rm P}}
\newcommand{\np}{{\rm NP}}

\newcommand{\elec}{\ensuremath{\cal E}}

\newtheorem{theorem}{Theorem}
\newtheorem{corollary}{Corollary}

\newcommand\qedblob{\ding{113}}
\def\literalqed{{\ \nolinebreak\hfill\mbox{\qedblob\quad}}}
\def\qed{\literalqed}

\showhyphens{}

\newcommand{\score}[1]{\ensuremath{{\rm score}(#1)}}

\newcommand{\prob}[3]{
\begin{description}
  \item[Name:] #1
  \item[Given:] #2
  \item[Question:] #3
\end{description}}

\begin{document}
\sloppy

\title{Complexity of Conformant Election Manipulation}

\author{Zack Fitzsimmons\\
 Dept.\ of Math.\ and Computer Science\\
 College of the Holy Cross\\
 Worcester, MA 01610 \and
  Edith Hemaspaandra\\
  Department of Computer Science\\
  Rochester Institute of Technology \\
  Rochester, NY 14623}%

\date{July 21, 2023}
\maketitle

\begin{abstract}
It is important to study how strategic agents can affect the outcome of an election. There has been a long line of research in the computational study of elections on the complexity of manipulative actions such as manipulation and bribery. These problems model scenarios such as voters casting strategic votes and agents campaigning for voters to change their votes to make a desired candidate win. A common assumption is that the preferences of the voters follow the structure of a domain restriction such as single peakedness, and so manipulators only consider votes that also satisfy this restriction. We introduce the model where the preferences of the voters define their own restriction and strategic actions must ``conform'' by using only these votes. In this model, the election after manipulation will retain common domain restrictions. We explore the computational complexity of conformant manipulative actions and we discuss how conformant manipulative actions relate to other manipulative actions.
\end{abstract}

\section{Introduction}

The computational study of election problems is motivated by the
utility of elections to aggregate preferences in multiagent systems
and to better understand the computational tradeoffs between different rules.
A major direction in this area %
has been to study the computational complexity of manipulative
actions on elections (see, e.g., Faliszewski and Rothe~\shortcite{fal-rot:b:handbook-comsoc-control-and-bribery}).

The problems of manipulation~\cite{bar-tov-tri:j:manipulating} and bribery~\cite{fal-hem-hem:j:bribery}
in elections represent two important ways that agent(s) can strategically affect the outcome of an election. Manipulation models the actions of a collection of strategic voters who seek to ensure that their preferred candidate wins by casting strategic votes.
Bribery models the actions of an agent, often referred to as the briber, who sets the votes of a subcollection of the voters to ensure that the briber's preferred candidate wins.
These problems each relate nicely to real-world scenarios
such as how voters may attempt to work together to strategically vote, or the actions of a campaign manager looking to influence the preferences of a group of voters to ensure their candidate wins.

In the manipulation problem each manipulator can cast any strategic vote, and similarly for the bribery problem the votes can be set to any collection of preferences.
However, this is not always a reasonable assumption to make.
Voters may have preferences that satisfy a domain restriction
such as single-peaked preferences~\cite{bla:j:rationale-of-group-decision-making} or single-crossing preferences~\cite{mir:j:single-crossing} where the manipulator or the briber are restricted to using votes that also satisfy the restriction.
We introduce new models of manipulation and bribery where the votes cast by
the manipulators or set by the briber must have already been stated in the election, i.e., the
votes must {\em conform} to the views of the electorate.
In these conformant models of manipulation and bribery
the election after the given manipulative action will retain common domain restrictions such
as being single-peaked or single-crossing.

We consider how the computational complexity of our conformant
models of manipulation and bribery compare to the standard models. Specifically, we show that there are settings where manipulation and bribery are easy in the standard model, but computationally difficult in the conformant model, and vice versa. This shows that there is no reduction in either direction between the standard and conformant cases (unless $\p = \np$).
Conformant manipulation and bribery are also each related to electoral control. The study of electoral control was introduced by Bartholdi, Tovey, and
Trick~\shortcite{bar-tov-tri:j:control}, and it models the actions of an agent who can modify the structure of an election to ensure a preferred candidate wins (e.g., by adding
or deleting voters).
We explore the connection between
conformant manipulation and the exact variant of voter control
as well as conformant bribery and the model of control by replacing voters introduced by Loreggia et al.~\shortcite{lor-nar-ros-bre-wal:c:replacing-control}. This
includes showing reductions between these problems as well as
showing when such reductions cannot exist.
Exact versions of electoral control problems can model scenarios where
the election chair seeks to ensure their preferred outcome by adding exactly the number of voters required to meet the quorum for a vote. This is in line with the standard motivation for control by replacing voters which includes settings such as voting in a parliament where the chair may replace some of the voters, but makes sure to leave the total number the same to avoid detection~\cite{lor-nar-ros-bre-wal:c:replacing-control}.

Our main contributions are as follows.
\begin{itemize}
    \item We introduce the problems of conformant manipulation and conformant bribery, which model natural settings for manipulative attacks on elections.
    
    \item We show that there is no
    reduction in either direction between the standard and conformant versions of
    manipulation (Theorems~\ref{thm:cm-vs-manip} and~\ref{thm:manip-fixed}) and between the standard and conformant versions of bribery (Theorem~\ref{thm:3veto-conformbribery} and Corollary~\ref{cor:bribe-vs-cb}) (unless $\p = \np$).
    
    \item We show conformant
    manipulation reduces to exact
    control by adding voters (Theorem~\ref{thm:restricted-to-exact}),
     conformant bribery reduces to control by replacing voters (Theorem~\ref{thm:bribery-to-replacing}),
    and that reductions do not exist
    in the other direction (Theorems~\ref{thm:xccav-vs-manip} and~\ref{thm:ccrv-vs-bribe})
    (unless $\p = \np$).
\item We obtain a trichotomy theorem for the complexity of exact control by adding voters problem (Theorem~\ref{thm:trichotomy}) for the important class of pure scoring rules.
\end{itemize}

\section{Related Models}

Since the conformant model of manipulation uses only the votes stated in the initial election, it is
related to the possible winner problem with uncertain weights introduced by Baumeister et al.~\shortcite{bau-roo-rot-sch-xia:c:weighted-pw}, which was recently extended by Neveling et al.~\shortcite{nev-rot-wei:c:possible-winner}.
In this problem weights of %
voters
are initially unset and it asks if there exists a way to set the weights such that a given candidate is a winner. As Baumeister et al.~\shortcite{bau-roo-rot-sch-xia:c:weighted-pw} mention, this generalizes control by adding/deleting voters (rather than manipulation or bribery).

There are other models for manipulative actions that have a similar motivation to our conformant
models, i.e., to have the manipulative action
not stand out.
Examples are bribery restrictions where the briber cannot put the
preferred candidate first in the bribed votes (negative bribery~\cite{fal-hem-hem:j:bribery}), restrictions on how much
the votes of voters are changed~\cite{yan-shr-guo:j:bribery-distance},
and restrictions on which voters can be bribed~\cite{dey-mis-nar:j:frugal-bribery}.
However, it is easy to see that in these models all sorts of new votes can be used by the briber, not just votes appearing in the initial election.

Another model in which the votes of the strategic agents are restricted is that of manipulative actions on restricted domains such as
single-peaked~\cite{bla:j:rationale-of-group-decision-making} and single-crossing preferences~\cite{mir:j:single-crossing}.
For example, for manipulation of a single-peaked election the manipulators must all cast votes that are single-peaked with respect to the rest of the electorate~\cite{wal:c:uncertainty-in-preference-elicitation-aggregation}. Notice that in the conformant models we
typically keep the domain restriction,
but manipulative actions for domain
restrictions are quite different
since in those settings the manipulators can cast a vote not stated by any of the nonmanipulators as long as it
satisfies the given restriction.

\section{Preliminaries}

An election $(C,V)$ consists of a 
set of candidates $C$ and a
collection of voters $V$. Each voter $v \in V$ has a corresponding vote, which is a strict total
order preference over the set of candidates. 

A voting rule, $\elec$, is a mapping from an election
to a subset of the candidate set referred to as the winner(s).

\subsection{Scoring Rules}

Our results focus on
(polynomial-time uniform) pure scoring rules.
A scoring rule is a voting rule defined by a family of scoring vectors of the form
$\langle \alpha_1, \alpha_2, \dots, \alpha_m \rangle$ with $\alpha_i \ge \alpha_{i+1}$ such that for a given election with $m$ candidates the
$m$-length scoring vector is used and each candidate receives $\alpha_i$
points for each vote where they are ranked $i$th. The candidate(s)
with the highest score win. 
We use the notation $\score{a}$ to denote the score
of a candidate $a$ in a given election.
Important examples of scoring rules are
 \begin{itemize}
    \item $k$-Approval, $\langle \underbrace{1,\dots,1}_k, 0, \dots, 0\rangle$
    \item $k$-Veto, $\langle 0, \dots, 0, \underbrace{-1, \dots, -1}_k\rangle$
    \item Borda, $\langle m-1, m-2, \dots, 1,0\rangle$
    \item First-Last, $\langle 1, 0, \ldots, 0, -1\rangle$\footnote{We will see that this rule exhibits very unusual complexity behavior. This rule has also been referred to as ``best-worst'' in social choice (see, e.g.,~\cite{kur:j:best-worst}).}
\end{itemize}
Note that there are uncountably many scoring rules, that scoring vectors may not be computable, and that the definition of scoring rule does not require any relationship between the scoring vectors for different numbers of candidates. To formalize
the notion of a natural scoring rule, we use the notion of \emph{(polynomial-time uniform) 
pure scoring rules}~\cite{bet-dor:j:possible-winner-dichotomy}. These are families of scoring rules where the scoring vector for $m+1$ candidates can be obtained from the scoring vector for $m$ candidates by adding one coefficient and for which there is a polynomial-time computable function that outputs, on input $0^m$, the scoring vector for $m$. Note that the election rules above are all pure
scoring rules. Also note that manipulative action problems for pure scoring rules are in \np.

\subsection{Manipulative Actions}

Two of the most commonly-studied manipulative
actions on elections are manipulation~\cite{bar-tov-tri:j:manipulating}
and bribery~\cite{fal-hem-hem:j:bribery}. 
We consider conformant variants of these standard
problems by requiring that the strategic votes cast by the manipulators or set by the briber must have appeared in the initial election. We define these problems formally
below.

\prob{\elec-Conformant Manipulation}%
{An election $(C,V)$,
a collection of manipulative voters $W$, and a preferred candidate $p$.}%
{Does there exist a way to set the votes of the manipulators
in $W$ using only the votes that occur in $V$ such that $p$
is an \elec-winner of the election $(C, V \cup W)$?}

\prob{\elec-Conformant Bribery}%
{An election $(C,V)$, a bribe limit $k$, and a preferred candidate $p$.}%
{Does there exist a way to set the preferences of a
subcollection of at most $k$ voters in $V$
to preferences in $V$ such that $p$ is an \elec-winner?}

\subsection{Computational Complexity}
\label{sec:cc}

We assume the reader is familiar with the complexity classes \p\ and
\np, polynomial-time many-one reductions, and what it means to be complete for a given class. 
Our \np-completeness proofs utilize
reductions from the well-known \np-complete problem 3-Dimensional Matching~\cite{kar:b:reducibilities}.

\prob{3-Dimensional Matching (3DM)}%
{Pairwise disjoint sets $X$, $Y$, and $Z$ such that $\|X\|=\|Y\|=\|Z\|=k > 0$
and set ${\cal M} \subseteq X \times Y \times Z$.}%
{Does there exist a ${\cal M}' \subseteq {\cal M}$ of size $k$ such that each $a \in X \cup Y \cup Z$ appears exactly once in ${\cal M}'$.}
3-Dimensional Matching remains \np-complete when each element $a \in X \cup Y \cup Z$
appears in exactly three triples (Garey and Johnson~\cite{gar-joh:b:int} show this for at most three triples, which can be adapted to exactly three by the approach from Papadimitriou and Yannakakis~\cite{pap-yan:j:spanning-tree}). Note that in that case $\|{\cal M}\| = 3k$.
For our polynomial-time algorithms, we will reduce to polynomial-time computable (edge) matching problems.
We define the most general version we use, Max-Weight $b$-Matching for Multigraphs, below. The version of this problem for simple graphs was shown to be in \p\ by Edmonds and Johnson~\cite{edm-joh:c:matching}, and as explained in~\cite[Section 7]{ger:b:matching}, it is easy to reduce such problems to 
Max-Weight Matching, which is well-known to be in P~\cite{edm:j:matching}, using the construction from~\cite{tut:j:factor}. (Note that we can assume that the $b$-values are bound by the number of edges in the graph.)

\prob{Max-Weight $b$-Matching for Multigraphs}%
{An edge-weighted multigraph $G = (V,E)$,
a function $b: V \to \mathbb{N}$, and
integer $k \geq 0$.}%
{Does there exist an
$E' \subseteq E$ of weight at least $k$ such
that each vertex $v \in V$ is incident to at
most $b(v)$ edges in $E'$?}

In addition to NP-hardness and polynomial-time results, we have results that link the complexity of voting problems to the complexity of
Exact Perfect Bipartite Matching~\cite{pap-yan:j:spanning-tree}.

\prob{Exact Perfect Bipartite Matching}%
{A bipartite graph $G = (V,E)$, a set of red edges $E' \subseteq E$, and integer $k \geq 0$.}%
{Does $G$ contain a perfect matching that contains exactly $k$ edges from $E'$?}

This problem was shown to be in \rp\ by Mulmuley, Vazirani, and Vazirani~\cite{mul-vaz-vaz:j:matching}, but it is a 40-year-old open problem whether it is in \p.

\section{Conformant Manipulation}
\label{s:manipulation}

The problem of manipulation asks if it is possible for a given collection
of manipulative voters to set their votes so that their preferred candidate wins.
This problem was first studied computationally by Bartholdi, Tovey,
and Trick~\shortcite{bar-tov-tri:j:manipulating} for the case of one manipulator and
generalized by Conitzer, Sandholm, and Lang~\shortcite{con-lan-san:j:few-candidates}
for the case of a coalition of manipulators.

In our model of conformant manipulation
the manipulators can only cast votes that at least one 
nonmanipulator has stated. As mentioned in the introduction,
this is so that the manipulators vote realistic preferences for the
given election by conforming to the preferences already stated.
Since this modifies the standard model of manipulation, we will
consider how the complexity of these problems relate to one another.

Manipulation is typically easy
for scoring rules. Hemaspaandra and Schnoor~\cite{hem-sch:c:dichotomy-two}
showed that manipulation is in \p\ for every pure scoring rule with a constant number of different coefficients. 
However, our model of conformant manipulation is \np-complete
for even the simple rule 4-approval. 
To show hardness, we use a construction similar to the construction that shows hardness for control by adding voters from Lin~\shortcite{lin:c:manip-k-app}.
\begin{theorem}\label{thm:cm-vs-manip}
4-Approval Conformant Manipulation is \np-complete.
\end{theorem}

\begin{proof}
Let the pairwise disjoint sets $X$, $Y$, and $Z$ such that $\|X\|=\|Y\|=\|Z\|=k > 0$ and ${\cal M} \subseteq X \times Y \times Z$ be an instance of 3DM
where each $a \in X \cup Y \cup Z$ appears in exactly three triples. Note that $\|{\cal M}\| = 3k$.
We
construct an instance of conformant manipulation as follows.

Let the candidate set $C$ consist of preferred candidate $p$, and for each $a$ in $X \cup Y \cup Z$, we have candidate $a$ and three padding candidates $a_1, a_2$, and $a_3$. 
We now construct the collection of nonmanipulators.

\begin{itemize}
\item For each $(x,y,z)$ in ${\cal M}$, we have a voter voting 
$p > x > y > z > \cdots$. 
\item For each $a$ in $X \cup Y \cup Z$, we have $4k-4$ voters voting
$a > a_1 > a_2 > a_3 > \cdots $.
\end{itemize}
We have $k$ manipulators.

Note that we have the following scores from the nonmanipulators. 
$\score{p} = 3k$ and for $a \in X \cup Y \cup Z$, $\score{a} = 4k - 4 + 3 = 4k - 1$ and $\score{a_1} = \score{a_2} = \score{a_3} = 4k - 4$.

If there is a matching, let the $k$ manipulators vote corresponding to the matching. Then $p$'s score increases by $k$, for a total of $4k$, and for each $a \in X \cup Y \cup Z$, $\score{a}$ increases by 1 for a total of $4k$. The scores of the dummy candidates remain unchanged. Thus, $p$ is a winner.

For the converse, suppose the manipulators vote such that $p$ is a winner. Since $k > 0$, after manipulation there is a candidate $a$ in $X \cup Y \cup Z$ with score at least $4k$. The highest possible score for $p$ after manipulation is $4k$, and this happens only if $p$ is approved by every manipulator.
It follows that every manipulator approves $p$ and for every $a$ in $X \cup Y \cup Z$, $a$ is approved by at most one manipulator. This implies that the votes of the manipulators correspond to a cover.
\qed\end{proof}

We just saw a case where the complexity of conformant manipulation is harder than
the standard problem (unless $\p = \np$). This is not always the case.
One setting where it is clear to see how to determine if conformant manipulation is possible
or not is when there are only a fixed number of manipulators.
We have only a polynomial number of votes to choose from (the votes of the nonmanipulators)
and so a fixed number of manipulators can brute force
these choices in polynomial time as long as determining the winner can be done in 
polynomial time.

\begin{theorem}\label{thm:manip-fixed}
Conformant Manipulation is in \p\ for every voting rule with a polynomial-time winner 
problem when there are a fixed number of manipulators.
\end{theorem}

This behavior is in contrast to what can occur for the standard model of
manipulation. One well-known example is that manipulation for the Borda rule is \np-complete
even when there are only two manipulators~\cite{dav-kat-nar-wal-xia:j:borda-manip,bet-nei-woe:c:board-manip}.
Intuitively the hardness
of manipulation in this case is realized by the choice of different votes that the manipulator(s) have.

\begin{corollary}
\label{thm:manip-vs-cm}
For Borda, Manipulation with two manipulators is \np-complete, but Conformant Manipulation with two manipulators is in \p.
\end{corollary}

In some ways conformant manipulation acts more like the problem of 
electoral control introduced by Bartholdi, Tovey, and 
Trick~\shortcite{bar-tov-tri:j:control}, specifically {\em Control by Adding 
Voters}, which asks when given an election, a collection
of unregistered voters, add limit $k$, and preferred candidate $p$, if $p$ is a
winner of the election after adding at most $k$ of the unregistered voters. In conformant manipulation we can think of the nonmanipulative voters as describing
the different votes to choose from for the manipulators.

At first glance it may appear that there is a straightforward reduction
from conformant manipulation to control by adding voters, but in
conformant manipulation {\em all $k$} of the manipulators must cast a vote, while
in control by adding voters {\em at most $k$} votes are added. In this way
conformant manipulation is closer to the ``exact'' variant of control by adding
voters where {\em exactly $k$} unregistered voters must be added.
The following is immediate.

\begin{theorem}\label{thm:restricted-to-exact}
Conformant Manipulation polynomial-time many-one reduces to
Exact Control by Adding Voters. 
\end{theorem}

Below we show that conformant manipulation is in \p\ for the voting rule 3-approval, but
exact control by adding voters is \np-complete.
And so there is no reduction from exact control by adding voters
to conformant manipulation (unless $\p = \np$).

\begin{theorem}
\label{thm:xccav-vs-manip}
For 3-Approval, Exact Control by Adding Voters is \np-complete, but
Conformant Manipulation is in \p.
\end{theorem}

\begin{proof}
Given an election $(C,V)$, $k$ manipulators, and a preferred
candidate $p$, we can determine if conformant manipulation is
possible in the following way.

If there is no nonmanipulator that
approves $p$ then $p$ is a winner if and only if there are no voters.

If there is at least one nonmanipulator that 
approves $p$, the manipulators will all cast votes
that approve of $p$, and so we know the final score of $p$
after manipulation: $fs_{p} = \score{p}+k$. However, these 
manipulator votes will each also approve of two other candidates. 
To handle this we can adapt the approach used by Lin~\shortcite{lin:c:manip-k-app} to show control by adding voters is in \p\ for 3-approval elections, which constructs a
reduction to Max-Cardinality $b$-Matching for Multigraphs.

For each candidate $a \neq p$, let $b(a) = fs_{p}-\score{a}$, i.e., the maximum number of approvals that $a$ can receive from the manipulators without beating $p$. If $b(a)$ is negative for any candidate $a \neq p$ then conformant manipulation is not possible.
For each distinct nonmanipulative vote of the form $\{p,a,b\} > \dots$, add $k$ edges between $a$ and $b$. Conformant manipulation is possible if and only if there is a $b$-edge matching of size at least $k$.
We now consider the complexity of exact control by adding voters.
Given an instance of 3-Dimensional Matching:
 pairwise disjoint sets $X$, $Y$, and $Z$ such that $\|X\|=\|Y\|=\|Z\|=k$, and ${\cal M} \subseteq X \times Y \times Z$ with ${\cal M} = \{M_1, \dots, M_{3k}\}$ we
construct the following instance of exact control by adding voters.

Let the candidate set $C = \{p,d_1,d_2\} \cup X \cup Y \cup Z$,
the preferred candidate be $p$, and the add limit be $k$.

Let there be one registered voter voting $p > d_1 > d_2 > \dots$,
and let the set of unregistered voters consist of
one voter voting $x > y > z > \dots$ for each $M_i = (x,y,z)$.

It is easy to see that $p$ can be made a winner by adding exactly
$k$ unregistered voters if and only if there is a matching of size $k$.\end{proof}
The standard case of control by adding voters for 3-approval is in \p~\cite{lin:c:manip-k-app}, but as shown above the exact case is \np-complete.
Related work that mentions exact control has results
only where the exact variant is also easy~\cite{erd-nev-reg-rot-yan-zor:j:towards-completing,fit-hem:c:random-classification}.
Note that, as observed in~\cite{fit-hem:c:random-classification}, control polynomial-time reduces to exact control, since, for example, 
$p$ can be made a winner by adding at most $k$ voters if and only if $p$ can be made a winner by adding 0 voters or 1 voter or 2 voters or \ldots), and so if exact control is easy, the standard case will be as well, and if the standard case is hard, then the exact case will be hard. Note that we are using a somewhat more flexible notion of reducibility than many-one reducibility here, 
since we are allowing the disjunction of multiple queries to the exact control problem. Such a reduction is called a {\em disjunctive truth-table (dtt) reduction}. This type of reduction is still much less flexible than a Turing reduction.

Is 3-approval special? For the infinite and natural class of pure scoring rules, Table~\ref{tbl:trichotomy} 
completely classifies the complexity of exact control by adding voters and compares this behavior to the complexity of control by adding voters~\cite{hem-hem-sch:c:dichotomy-one} and control by deleting voters~\cite{hem-sch:c:dichotomy-two}. In particular, and in contrast to earlier results, we obtain a \emph{trichotomy} theorem for exact control by adding voters.\footnote{\label{f:epbm}Exact Perfect Bipartite Matching~\cite{pap-yan:j:spanning-tree} is defined in
Section~\ref{sec:cc}.
As mentioned there the complexity of this problem is still open.
And so Theorem~\ref{thm:trichotomy} is a trichotomy theorem unless we solve a 40-year-old open problem.}

\begin{table}

\begin{center}
\begin{tabular}{l|c|c|c}
& {\p} & \ {eq.\ to EPBM} \ & \ {\np-complete} \ \\ \hline
Exact Control by & 0/1/2-approval, 1/2-veto & first-last & all other cases \\
Adding Voters & & &\\[0.5em]
\hline
Control by & 0/1/2/3-approval, 1/2-veto, & & all other cases \\ 
Adding Voters & first-last, $\langle \alpha, \beta, 0, \dots, 0 \rangle$ & & \\[0.5em] \hline
Control by & 0/1/2-approval, 1/2/3-veto, & & all other cases\\
Deleting Voters & first-last, $\langle 0, \dots, 0, -\beta, -\alpha \rangle$ & & \\ %
\end{tabular}
\end{center}

\caption{This table classifies the complexity of all pure scoring rules for the specified control action. 
A scoring rule entry represents all pure scoring rules that are ultimately equivalent to that scoring rule. The dichotomy for control by adding voters is due to~\cite{hem-hem-sch:c:dichotomy-one}, the dichotomy for  control by deleting voters to~\cite{hem-sch:c:dichotomy-two}, and the result for exact control by adding voters for first-last is due to~\cite{fit-hem:c:random-classification}. EPBM stands for the Exact Perfect Bipartite Matching, which is defined in 
Section~\ref{sec:cc}.}
\label{tbl:trichotomy}
\end{table}

\begin{theorem}\label{thm:trichotomy}
For every pure scoring rule $f$,
\begin{enumerate}
\item If $f$ is ultimately (i.e., for all but a finite number of candidates) equivalent to 0-approval, 1-approval, 2-approval, 1-veto, or 2-veto, 
exact control by adding voters is in \p.
\item If $f$ is ultimately equivalent to first-last, then exact control by adding voters is (dtt) equivalent to 
the problem Exact Perfect Bipartite Matching~\cite{pap-yan:j:spanning-tree}.
\item In all other cases, exact control by adding voters is \np-complete (under dtt reductions). 
\end{enumerate}
\end{theorem}

\begin{proof}~\begin{enumerate}
\item
The case for 0-approval is trivial, since all candidates are always winners. The 1-approval and 1-veto cases follow by straightforward greedy algorithms. The case for 2-approval can be found in~\cite{fit-hem:c:random-classification} and the case for 2-veto is similar.

\item
The case for first-last can be found in~\cite{fit-hem:c:random-classification}. 

\item
Note that the remaining cases are hard for control by adding voters or for control by deleting voters.
We have already explained how we can dtt reduce control by adding voters to exact control by adding voters.
Similarly, we can dtt reduce control by deleting voters to exact control by adding voters, since
$p$ can be made a winner by deleting at most $k$ voters if and only $p$ can be made a winner by adding to the empty set $n-k$ voters or $n-k+1$ voters or $n-k+2$ voters or \ldots or  $n$ voters, where $n$ is the total number of voters.
It follows that all these cases are NP-complete (under dtt reductions).
\end{enumerate}
\end{proof}

In this section, we compared conformant manipulation with manipulation and conformant manipulation with exact control
by adding voters. We can also look at conformant manipulation versus control by adding voters. Here we also find voting rules where the manipulative actions differ. In particular, control by adding voters for first-last is in \p~\cite{hem-hem-sch:c:dichotomy-one}, but we show in the appendix that conformant manipulation for first-last is equivalent to exact perfect bipartite matching.

\begin{theorem}
\label{thm:cm-vs-ccav}
First-Last Conformant Manipulation is equivalent to Exact Perfect Bipartite Matching (under dtt reductions).
\end{theorem}

We also show in the appendix that there exists a voting rule where control by adding voters is hard and conformant manipulation is easy. 

\begin{theorem}\label{t:ccav-vs-cm}
There exists a voting rule where
Control by Adding Voters is \np-complete, but Conformant Manipulation is in \p.
\end{theorem}

This shows that a reduction from control by adding voters to
conformant manipulation does not exist (unless $\p = \np$).

\section{Conformant Bribery}
\label{s:bribery}

We now turn to our model of conformant bribery.
The standard bribery problem introduced by
Faliszewski, Hemaspaandra, and Hemaspaandra~\shortcite{fal-hem-hem:j:bribery}
asks when given an election, a bribe limit $k$, and
a preferred candidate $p$, if there exists a subcollection of at most $k$ voters whose votes can be changed such that $p$ is a winner. In our model of conformant bribery the votes can only be changed to votes that appear in the initial election.
As with conformant manipulation, this is so that the votes are changed to preferences that are still realistic with respect to the preferences already stated.
Notice how this also nicely models how a voter can convince another to vote their same vote.
In the same way as we did with manipulation in the previous section, we can compare the complexity behavior of
our conformant model with respect to the standard model.

Bribery is in \p\ for the voting rule 3-veto~\cite{lin:c:manip-k-app},
but we show below that our model of conformant
bribery is \np-complete for this rule.
\begin{theorem}\label{thm:3veto-conformbribery}
3-Veto Conformant Bribery is \np-complete.
\end{theorem}

\def\threevetobribery{%
\begin{proof}
Let $X$, $Y$, and $Z$ be pairwise disjoint sets such that $\|X\|=\|Y\|=\|Z\|=k$, and ${\cal M} \subseteq X \times Y \times Z$ with ${\cal M} = \{M_1, \dots, M_{3k}\}$ be an instance of 3DM where each element $a \in X \cup Y \cup Z$ appears in exactly three triples.
We construct an instance of conformant bribery as follows.

Let the candidate set $C = \{p\} \cup X \cup Y \cup Z \cup
\{p_1,p_2\}$. Let $p$ be the preferred candidate and let $k$ be the bribe limit.
Let there be the following voters.
\begin{itemize}
    \item For each $M_i = (x,y,z)$,
    \begin{itemize}
        \item One voter voting $\cdots > x > y > z$
    \end{itemize}
    \item $k+4$ voters voting $\cdots > p > p_1 > p_2$
\end{itemize}
We view the corresponding scoring vector
for $3$-veto as $\langle 0, \dots, 0, -1, -1, -1\rangle$ to
make our argument more straightforward. And so,
before bribery $\score{p} = \score{p_1} = \score{p_2} = -k-4$ and for each $a \in X \cup Y \cup Z$, $\score{a} = -3$.

If there exists a matching ${\cal M}' \subseteq {\cal M}$
of size $k$, for each $M_i \in {\cal M}'$ such that $M_i = (x,y,z)$ we
can bribe one of the voters voting $\cdots > p > p_1 > p_2$ to vote $\cdots > x > y > z$. Since ${\cal M}'$
is a matching the score of each candidate $a \in X \cup Y \cup Z$ decreases by 1 to be $-4$, and since
$k$ of the voters vetoing $p$ are bribed the score of $p$ increases by $k$ to $-4$ and $p$ is a winner.

For the converse, suppose there is a successful conformant bribery. Only the voters vetoing $p$ should be bribed,
and so the score of $p$ after bribery is $-4$. The score of each candidate $a \in X \cup Y \cup Z$
must decrease by at least 1, and so it is easy to see that a successful conformant
bribery of at most $k$ voters must correspond to a perfect matching.\end{proof}}

\threevetobribery

We now consider the case where bribery is hard for the standard model, but easy in our conformant model.
Since the briber is restricted to use only votes that appear in the initial election we have the same behavior
as stated in Theorem~\ref{thm:manip-fixed} for conformant manipulation.

\begin{theorem}\label{thm:bribe-fixed}
Conformant Bribery is in \p\ for every voting rule with a polynomial-time winner 
problem when there is a fixed bribe limit.
\end{theorem}

There are generally fewer results looking at a fixed bribe limit than there are looking at a fixed number of manipulators.
One example is that for Single Transferable Vote (STV), bribery is \np-complete even when the bribe limit is 1~\cite{xia:c:margin-of-victory}, but
our focus is on scoring rules.
Brelsford at al.~\shortcite{bre-fal-hem-sch-sch:c:approximating-elections} show bribery is
\np-complete for Borda, but do not consider a fixed bribe limit. However, it is easy to adapt the \np-hardness proof for Borda manipulation with two manipulators from Davies et al.~\shortcite{dav-kat-nar-wal-xia:j:borda-manip}. The main idea is to add two voters that are so bad for the preferred candidate that they have to be bribed.

\begin{theorem}\label{thm:bordabribery}
Borda Bribery with a bribe limit of 2 is \np-complete.
\end{theorem}

\def\bordabribery{%
\begin{proof}
We need to following properties of the instance of Borda Manipulation with two manipulators constructed in Davies et al.~\shortcite{dav-kat-nar-wal-xia:j:borda-manip}. (For this
proof we use the notation from Davies et al.~\cite{dav-kat-nar-wal-xia:j:borda-manip}.)
The constructed election has a collection of (nonmanipulative) voters $V$ and $q+3$ candidates. Preferred candidate $p$ scores $C$ and candidate $a_{q+1}$ scores $2(q+2)+C$. This implies that in order for $p$ to be a winner, the two manipulators must vote $p$ first and $a_{q+1}$ last. We add two voters voting $a_{q+1} > \cdots > p$. Note that such votes are very bad for $p$ and that we have to bribe two voters voting $a_{q+1} > \cdots > p$. In order to ensure that we bribe exactly the two added voters, it suffices to observe that we can ensure in the construction from Davies et al.~\shortcite{dav-kat-nar-wal-xia:j:borda-manip} that $p$ is never last in any vote in $V$. So, $p$ can be made a winner by bribing two voters in $V \cup \{a_{q+1} > \cdots > p, a_{q+1} > \cdots > p\}$ if and only if $p$ can be made a winner by two manipulators in the election constructed by Davies et al.~\shortcite{dav-kat-nar-wal-xia:j:borda-manip}.\end{proof}
}

\bordabribery

\begin{corollary}
\label{cor:bribe-vs-cb}
For Borda, Bribery is \np-complete with a bribe limit of 2, but Conformant Bribery is in \p\ with a bribe limit of 2.
\end{corollary}

Bribery can be thought of as control by deleting voters followed by manipulation.
For conformant bribery we can see that the same will hold, just with conformant manipulation.
However, we also have a correspondence to the problem of control by replacing voters introduced by Loreggia et
al.~\shortcite{lor-nar-ros-bre-wal:c:replacing-control}. Control by replacing voters asks when given an election, a collection of
unregistered voters, parameter $k$, and preferred candidate $p$, if $p$ can be made a winner by replacing at most $k$ voters in
the given election with a subcollection of the unregistered voters of the same size.
It is straightforward to reduce conformant bribery to control by replacing voters (for each 
original voter $v$, we have a registered copy of $v$ and $k$ unregistered copies of $v$), and so we inherit polynomial-time results from this setting.

\begin{theorem}\label{thm:bribery-to-replacing}
Conformant Bribery polynomial-time many-one reduces to
Control by Replacing Voters.
\end{theorem}

It's natural to ask if there is a setting where conformant bribery is easy, but control by replacing voters is hard, and we show in the appendix
that this is
in fact the case.

\begin{theorem}\label{thm:ccrv-vs-bribe}
There exists a voting rule where
Control by Replacing Voters is \np-complete, but Conformant Bribery is in \p.
\end{theorem}

\def\ccavversusbribery{%
\begin{proof}
We will use the voting rule defined in the proof of Theorem~\ref{t:ccav-vs-cm}.

We first show that conformant bribery is in P. If $k=0$ or if there is at most one voter,
$p$ can be made a winner if and only if $p$ is already a winner.
If no voter has $p$ first, then the score of $p$ after bribery is 0.
The total score over all candidates is $n((m+3)-3)=nm$, where $n$ is the number of voters, and
so there is a candidate scoring at least $n > 0$ and so $p$ is not a winner.
If there is a voter that has $p$ first, $k > 0$, and there is more than one voter, we will bribe voters to have $p$ first. 
Since p will be first at least twice, we are in case (1), which is basically plurality and so conformant bribery is in \p\ (greedily bribe a voter voting for a highest-scoring candidate that is not $p$ to vote for $p$ until we have bribed $k$ voters or until all voters vote for $p$).

To show that control by replacing voters is hard, we pad the \np-hardness reduction for 3-veto control by adding voters from
Lin~\shortcite{lin:c:manip-k-app}.

Let the pairwise disjoint sets $X = \{x_1, \ldots, x_k\}$, $Y = \{y_1, \ldots, y_k\}$, and $Z = \{z_1 \ldots, z_k\}$, and ${\cal M} \subseteq X \times Y \times Z$ with ${\cal M} = \{M_1, \dots, M_{3k}\}$ be an instance of 3DM.
We
construct an instance of control by replacing voters as follows. The candidates are
$p$, every element of $X \cup Y \cup Z$, and dummies $d_1, \ldots, d_{3k}$. Without loss of generality, assume that $k \geq 1$. Let $m$ be the number of candidates.
We have the following $4k+1$ registered voters. 
\begin{itemize}
\item 
$p > \cdots > d_1 > d_2 > d_3$.
\item two voters that vote
$x_i > \cdots > d_{3i-2} > d_{3i-1} > d_{3i}$ for $1 \leq i \leq k$.
\item
$y_i > \cdots > d_{3i-2} > d_{3i-1} > d_{3i}$ for $1 \leq i \leq k$.
\item
$z_i > \cdots > d_{3i-2} > d_{3i-1} > d_{3i}$ for $1 \leq i < k$.
\item $z_k > \cdots > d_1 > d_2 > p$.
\end{itemize}
We have the following $3k$ unregistered voters. For $M_i = (x,y,z)$, one voter voting $d_i > \cdots > x > y > z$. Our bound is $k$.

Note that there is only one voter (registered or unregistered) with $p$ first and so in order for $p$ to be a winner, each candidate can be first at most once. This implies that we need to replace a voter that has $x_i$ first for each $x_i$. Since the replacing
limit is $k$ and for each $i, 1 \leq i \leq k$, the two votes that have $x_i$ first are identical, this fixes the votes that need to be replaced.

We will now look at the scores of the candidates from the registered voters after deleting the $k$ votes that need to be replaced.

The score of $p$ is $m+2$, the score of each $a \in X \cup Y \cup Z$ is $m+3$, and the score of each $d_i$ is less than $-1$. Note that we are always in case 2, since no candidate will be first more than once.
If there is a matching of size $k$, adding the voters corresponding to the matching (as replacements for the deleted voters) will make $p$ a winner. And to make $p$ a winner we need to make sure that the score of every candidate in $X \cup Y \cup Z$ goes down by at least 1 point. This implies that the replacement voters correspond to a cover. And since we add $k$ replacement voters to make this happen, this is a perfect matching.\end{proof}
}

In the related work on control by replacing voters,
only the complexity for 2-approval remained open (see~Erd{\'e}lyi et al.~\shortcite{erd-nev-reg-rot-yan-zor:j:towards-completing}). This was recently
shown to be in \p\ by Fitzsimmons and Hemaspaandra~\cite{fit-hem:t:weighted-matching}.
This result immediately implies that
conformant bribery for 2-approval is also in \p.

\begin{theorem}\label{thm:2-app-conformant-bribery}
2-Approval Conformant Bribery is in \p.
\end{theorem} 

2-approval appears right at the edge of what is easy. 
For 3-approval, control by deleting voters and bribery are hard~\cite{lin:c:manip-k-app},
control by replacing voters is hard~\cite{erd-nev-reg-rot-yan-zor:j:towards-completing}, and we show in the appendix
that conformant bribery is hard as well (recall that for 3-approval, conformant manipulation (Theorem~\ref{thm:xccav-vs-manip}) and control by adding voters~\cite{lin:c:manip-k-app} are easy, essentially because all we are doing is ``adding'' votes that approve $p$).

\begin{theorem}\label{thm:3app-conform-bribery}
3-Approval Conformant Bribery is \np-complete.
\end{theorem}

\def\3approvalconformbribery{%
\begin{proof}
Let pairwise disjoint sets $X = \{x_1, \dots, x_k\}$, $Y = \{y_1, \dots, y_k\}$, and $Z = \{z_1, \dots, z_k\}$, and $M \subseteq X \times Y \times Z$ with ${\cal M} = \{M_1, \dots, M_{3k}\}$ be an instance of 3-Dimensional Matching where each element $a \in X \cup Y \cup Z$ appears in exactly three triples.
We construct an instance of conformant bribery as follows.

Let the candidate set $C = \{p\} \cup X \cup Y \cup Z \cup
\{p_1,p_2\} \cup \{d_1,\dots,d_{6k(k-1)}\}$. Let $p$ be the preferred candidate and let $k$ be the bribe limit.
Let there be the following voters.
\begin{itemize}
    \item For each $M_i = (x,y,z)$,
    \begin{itemize}
        \item One voter voting $x > y > z > \cdots$
    \end{itemize}
    \item One voter voting $p > p_1 > p_2 > \cdots$
\end{itemize}
From the votes above each candidate in $X \cup Y \cup Z$
has a score of 3, and $p$ and the two dummy candidates $p_1$
and $p_2$ have a score of 1. The vote that approves of $p$
will be the one used by the briber to set the bribed votes.
We now pad the election with $3k(k-1)$ additional votes such that for
each $a \in X \cup Y \cup Z$ there are $k-1$ votes that approve
of $a$ 
and two of the dummy candidates $d$ (two different dummies for each vote).
After this additional padding the score of each of the candidates
in $X \cup Y \cup Z$ is $k+2$.

If there is a perfect matching ${\cal M}' \subseteq {\cal M}$, then for each $M_i = (x,y,z) \in {\cal M}'$ bribe the voters
voting $x > y > z > \cdots$ to vote $p > p_1 > p_2 > \cdots$ The score of each candidate in $X \cup Y \cup Z$ %
decreases by exactly 1 and the score of $p$ increases by $k$. Therefore $p$ is a winner.

Suppose there is a conformant bribery such that $p$ wins. A successful conformant bribery must
replace $k$ of the votes in the election with the vote $p > p_1 > p_2 > \cdots$
Thus the score after bribery for $p$ is $k+1$ and the bribery must decrease the score of each candidate in  $X \cup Y \cup Z$
by at least 1. It is clear that votes of the form $x > y > z > \cdots$ are bribed and that
they correspond to a perfect matching.\end{proof}
}

As a final note, we mention that for first-last, conformant bribery, like conformant manipulation, is equivalent to
exact perfect bipartite matching again showing the unusual
complexity behavior of this rule.

\begin{theorem}
\label{thm:cb-vs-ccav} 
First-Last Conformant Bribery is equivalent to Exact Perfect Bipartite Matching (under dtt reductions).
\end{theorem}

The proof of the above theorem has been deferred to the appendix.

\section{Conclusion} %

The conformant models of manipulation and bribery
capture a natural setting for election manipulation. We found that there is no reduction
between the standard and conformant models in either direction (unless $\p = \np$), and further explored the connection between these models and types of electoral control.

We found the first trichotomy theorem for scoring rules.
This theorem concerns the problem of exact control by adding voters and
highlights the unusual %
complexity behavior of the scoring rule first-last.
We show that this unusual complexity behavior also occurs for our conformant models.

We also observed interesting behavior for exact variants of control, including a nontrivial case where the complexity of a problem increases when going from the standard to the exact case.

We see several interesting directions for future work. For example, we could look at conformant versions for other 
bribery problems (e.g., priced bribery) or for restricted domains such as single-peakedness. We are also interested in further exploring the complexity landscape of problems for the scoring rule first-last.

\section*{Acknowledgements}

This work was supported in part by grant NSF-DUE-1819546.
We thank the anonymous reviewers for their helpful comments and suggestions.

\newcommand{\etalchar}[1]{$^{#1}$}

\appendix

\section{Appendix} %

\noindent
{\bf Theorem \ref{thm:cm-vs-ccav}.} \
{\em First-Last Conformant Manipulation is equivalent to Exact Perfect Bipartite Matching (under dtt reductions).}

\medskip

\begin{proof}
Conformant manipulation reduces to exact control by adding voters (Theorem~\ref{thm:restricted-to-exact}), which is equivalent to
exact perfect bipartite matching for first-last elections~\cite{fit-hem:c:random-classification} (under dtt reductions). This implies that for first-last, conformant manipulation reduces to exact perfect bipartite matching.

We now show the other direction.
Consider the following instance of exact perfect bipartite matching.
Let $G = (V, E)$ be a bipartite graph where the vertices can be partitioned into the sets $A = \{a_1, \dots, a_n\}$ and $B = \{b_1, \dots, b_n\}$ such that all edges are between vertices in $A$ and $B$, and let $k$ be the desired number of red edges in an exact perfect matching.
We construct an instance of conformant manipulation as follows.

Let the candidate set $C = A \cup B \cup D \cup \{p\}$
where $D$ consists of for each edge $(a_i,b_j)$ in $E(G)$, the dummy candidates $d_1^{ij}, \dots, d_{n+1}^{ij}$.
Let $p$ be the preferred candidate and let there be $n(n+1) + k$ manipulators. Let there be the following nonmanipulators.

\begin{description}
  \item[Vertex Voters:] For each $i, 1 \le i \le n$,
  \begin{itemize}
      \item One voter voting $a_i > \dots > b_i$.
  \end{itemize}

  \item[Edge Voters:] For each edge $(a_i,b_j)$ in $E(G)$,
  \begin{itemize}
      \item If the edge $(a_i,b_j)$ is red we have the following $n+3$ voters.

    \medskip
    \begin{tabular}{ccccc}
    $b_j$ & $>$ & $\cdots$ & $>$ & $d_1^{ij}$\\[2pt]
    $d_1^{ij}$ & $>$ & $\cdots$ & $>$ & $d_2^{ij}$\\[2pt]
    $d_2^{ij}$ & $>$ & $\cdots$ & $>$ & $d_3^{ij}$\\
    && $\vdots$ &&\\
    $d_{n}^{ij}$ & $>$ & $\cdots$ & $>$ &$d_{n+1}^{ij}$\\[2pt]
    $d_{n+1}^{ij}$ & $>$ & $\cdots$ &  $>$ & $a_i$\\[2pt]
    $a_i$ & $>$ & $\cdots$ & $>$ & $b_j$\\[2pt]
    \end{tabular}
    \medskip
    
    \item If the edge $(a_i,b_j)$ is not red we have the following $n+2$ voters.
    
    \medskip
    \begin{tabular}{ccccc}
    $b_j$ & $>$ & $\cdots$ & $>$ & $d_1^{ij}$\\[2pt]
    $d_1^{ij}$ & $>$ & $\cdots$ & $>$ & $d_2^{ij}$\\[2pt]
    $d_2^{ij}$ & $>$ & $\cdots$ & $>$ & $d_3^{ij}$\\
    && $\vdots$ &&\\
    $d_{n-1}^{ij}$ & $>$ & $\cdots$ & $>$ &$d_{n}^{ij}$\\[2pt]
    $d_{n}^{ij}$ & $>$ & $\cdots$ & $>$ & $a_i$\\[2pt]
    $a_i$ & $>$ & $\cdots$ & $>$ & $b_j$\\[2pt]
    \end{tabular}
    \medskip

  \end{itemize}
\end{description}

The vertex voters give each candidate in $B$ a score of $-1$, each candidate in $A$ a score of $1$, and all remaining candidates %
a score of 0, and each set of edge voters gives 0 points to all of the candidates. Since $p$ receives 0 points from each nonmanipulator, for $p$ to be a winner all candidates must have score 0 after
the manipulation.
We will show that there is a perfect bipartite matching
with exactly $k$ red edges if and only if
there is a way to set the votes of the $n(n+1)+k$ manipulators using the votes of the nonmanipulators such that $p$ is a
winner.

First, suppose there exists a perfect bipartite matching with exactly $k$ red edges. For each edge $(a_i,b_j)$ in the matching, have manipulators
cast all edge votes corresponding to the edge except for the vote $a_i > \dots > b_j$, and so the candidate $b_j$ gains one point and $a_i$ loses one point. For each red edge we use $n+2$ manipulators and for each nonred edge we use $n+1$
manipulators, which is a total of
$(n-k)(n+1) + k(n+2) = n(n+1)+k$
manipulators. After manipulation all candidates have
a score of 0 and $p$ is a winner.

For the converse, suppose there exists a successful
conformant manipulation. As mentioned above, all candidates must have a score of 0 after manipulation. Each of the candidates in $B$ must lose a point, each of the candidates in $A$ must gain a point, and each candidate in
$D$ should not gain or lose a point. This implies that for each candidate $b \in B$, there exists a candidate $a \in A$ such that we cast all edge votes corresponding to $(a,b)$ other than $a > \cdots > b$. So we need to cast at least $n(n+1) + \ell$ votes, where $\ell$ is the number of $(a,b)$ that are red. The remaining $k - \ell$ manipulators need to cast votes such that each candidate gets 0 points. Since $k < n+2$, this is only possible if $k - \ell = 0$, i.e., there is a perfect matching with exactly $k$ red edges.
\end{proof}

\bigskip

\noindent
{\bf Theorem \ref{t:ccav-vs-cm}.} \
{\em There exists a voting rule where
Control by Adding Voters is \np-complete, but Conformant Manipulation is in \p.}

\medskip

\begin{proof}
It is not so simple to find
a reasonable example here. For example, if we require that all votes are distinct, conformant manipulation with at least one manipulator will never be possible, but that rather abuses the model. A nontrivial example is the following.

Consider the following voting rule.
$c$ is a winner if and only if (1) $c$ is first at least twice and no other candidate is first more than $c$, or (2) all candidates are first at most once and $c$ is a winner using scoring rule $\langle m+3,0,\ldots, 0, -1, -1, -1\rangle$, where $m$ is the number of candidates. For this rule, Control by Adding Voters is \np-complete, but Conformant Manipulation is in \p.

We first show that conformant manipulation is in \p. If we have no manipulators, we simply evaluate. If we have manipulators but no nonmanipulators, there are no winners. So, assume we have manipulators and nonmanipulators. If there is a nonmanipulator that has $p$ first, then all manipulators should vote $p$ first. Since $p$ will be first at least twice, we are in case (1), which is basically plurality and so conformant manipulation in is P. If no nonmanipulator has $p$ first, then $p$ can not be first in the manipulators either. So, $\score{p}$ is at most 0. The total score over all candidates is $(n+k)((m+3)-3)$, where $n$ is the number of nonmanipulators and $k$ the number of manipulators, and so there is a candidate scoring at least $(n+k) > 0$ and so $p$ is not a winner.

To show that control by adding voters is hard, we pad the \np-hardness reduction for 3-veto control by adding voters from
Lin~\shortcite{lin:c:manip-k-app}.

Let the pairwise disjoint sets $X = \{x_1, \ldots, x_k\}$, $Y = \{y_1, \ldots, y_k\}$, and $Z = \{z_1 \ldots, z_k\}$, and ${\cal M} \subseteq X \times Y \times Z$ with ${\cal M} = \{M_1, \dots, M_{3k}\}$ be an instance of 3DM.
We
construct an instance of control by adding voters as follows. The candidates are
$p$, every element of $X \cup Y \cup Z$, and dummies $d_1, \ldots, d_{3k}$. Without loss of generality, assume that $k \geq 1$. Let $m$ be the number of candidates.
We have the following $3k+1$ registered voters. 
\begin{itemize}
\item 
$p > \cdots > d_1 > d_2 > d_3$.
\item
$x_i > \cdots > d_{3i-2} > d_{3i-1} > d_{3i}$ for $1 \leq i \leq k$.
\item
$y_i > \cdots > d_{3i-2} > d_{3i-1} > d_{3i}$ for $1 \leq i \leq k$.
\item
$z_i > \cdots > d_{3i-2} > d_{3i-1} > d_{3i}$ for $1 \leq i < k$.
\item $z_k > \cdots > d_1 > d_2 > p$.
\end{itemize}
So the score of $p$ is $m+2$, the score of each $a \in X \cup Y \cup Z$ is $m+3$, and the score of each $d_i$ is less than $-1$. Our adding bound is $k$. We have the following $3k$ unregistered voters. For $M_i = (x,y,z)$, one voter voting $d_i > \cdots > x > y > z$.
Note that we are always in case 2, since no candidate will be first more than once.
If there is a matching of size $k$, adding the votes corresponding to the matching will make $p$ a winner. And to make $p$ a winner we need to make sure that the score of every candidate in $X \cup Y \cup Z$ goes down by at least 1 point. This implies that the added votes correspond to a cover. And since we add at most $k$ voters to make this happen, this is a perfect matching.\end{proof}

\bigskip

\noindent
{\bf Theorem~\ref{thm:ccrv-vs-bribe}.} \
{\em There exists a voting rule where
Control by Replacing Voters is \np-complete, but Conformant Bribery is in \p.}

\medskip

\ccavversusbribery

\bigskip

\noindent
{\bf Theorem~\ref{thm:3app-conform-bribery}.} \
{\em 3-Approval Conformant Bribery is \np-complete.}

\medskip

\3approvalconformbribery

\noindent
{\bf Theorem \ref{thm:cb-vs-ccav}.} \
{\em First-Last Conformant Bribery is equivalent to Exact Perfect Bipartite Matching (under dtt reductions).}

\medskip

\begin{proof}
Conformant bribery reduces to control by replacing voters (Theorem~\ref{thm:bribery-to-replacing}), which is equivalent to
exact perfect bipartite matching for first-last elections~\cite{fit-hem:c:random-classification} (under dtt reductions). This implies that for first-last, conformant bribery reduces to exact perfect bipartite matching.

For the reduction from exact perfect bipartite matching, the approach used in the proof of Theorem~\ref{thm:cm-vs-ccav} can be adapted to show the above theorem by adding an additional two candidates $r$
and $\hat{r}$, $n(n+1)+k$ voters voting $r > \dots > \hat{r}$,
and setting the bribe limit to $n(n+1)+k$. All of these new voters must be bribed to vote the corresponding manipulation described in the proof of Theorem~\ref{thm:cm-vs-ccav}. It is
straightforward to see that the reduction holds.
\end{proof}

\end{document}